\documentclass[12pt]{article}

\def\UseBibLatex{1}

\makeatletter
\def\input@path{{styles/}}
\makeatother

\providecommand{\BibLatexMode}[1]{}
\providecommand{\BibTexMode}[1]{}

\renewcommand{\BibLatexMode}[1]{#1}
\renewcommand{\BibTexMode}[1]{}

\ifx\UseBibLatex\undefined%
  \renewcommand{\BibLatexMode}[1]{}
  \renewcommand{\BibTexMode}[1]{#1}
\fi

\BibLatexMode{%
   \usepackage[bibencoding=utf8,style=alphabetic,backend=biber]{biblatex}%
   \usepackage{sariel_biblatex}%
}

\usepackage[cm]{fullpage}%
\usepackage{amsmath}%
\usepackage{amssymb}%
\usepackage[table]{xcolor}%

\usepackage[amsmath,thmmarks]{ntheorem}%

\usepackage{titlesec}%
\usepackage{xcolor}%
\usepackage{mleftright}%
\usepackage{xspace}%
\usepackage{graphicx}
\usepackage{hyperref}%
\usepackage[inline]{enumitem}
\usepackage{hyperref}%
\usepackage[ocgcolorlinks]{ocgx2}

\hypersetup{%
      unicode,
      breaklinks,%
      colorlinks=true,%
      urlcolor=[rgb]{0.25,0.0,0.0},%
      linkcolor=[rgb]{0.5,0.0,0.0},%
      citecolor=[rgb]{0,0.2,0.445},%
      filecolor=[rgb]{0,0,0.4},
      anchorcolor=[rgb]={0.0,0.1,0.2}%
}

\titlelabel{\thetitle. }%

\theoremseparator{.}%

\theoremstyle{plain}%
\newtheorem{theorem}{Theorem}[section]

\newtheorem{lemma}[theorem]{Lemma}

\newtheorem{corollary}[theorem]{Corollary}

\theoremstyle{plain}%
\theoremheaderfont{\sf} \theorembodyfont{\upshape}%
\newtheorem*{remark:unnumbered}[theorem]{Remark}%
\newtheorem{remark}[theorem]{Remark}%

\newtheorem{defn}[theorem]{Definition}

\theoremheaderfont{\em}%
\theorembodyfont{\upshape}%
\theoremstyle{nonumberplain}%
\theoremseparator{}%
\theoremsymbol{\myqedsymbol}%
\newtheorem{proof}{Proof:}%

\providecommand{\emphind}[1]{}%
\renewcommand{\emphind}[1]{\emph{#1}\index{#1}}

\definecolor{blue25emph}{rgb}{0, 0, 11}

\providecommand{\emphic}[2]{}
\renewcommand{\emphic}[2]{\textcolor{blue25emph}{%
      \textbf{\emph{#1}}}\index{#2}}

\providecommand{\emphi}[1]{}%
\renewcommand{\emphi}[1]{\emphic{#1}{#1}}

\definecolor{almostblack}{rgb}{0, 0, 0.3}

\providecommand{\emphw}[1]{}%
\renewcommand{\emphw}[1]{{\textcolor{almostblack}{\emph{#1}}}}%

\providecommand{\emphOnly}[1]{}%
\renewcommand{\emphOnly}[1]{\emph{\textcolor{blue25emph}{\textbf{#1}}}}

\newcommand{\myqedsymbol}{\rule{2mm}{2mm}}
\newcommand{\SarielThanks}[1]{%
   \thanks{%
      Department of Computer Science; %
      University of Illinois; %
      201 N. Goodwin Avenue; %
      Urbana, IL, 61801, USA; %
      \href{mailto:spam@illinois.edu}{sariel@illinois.edu}; %
      \url{http://sarielhp.org/}. %
   #1%
   }%
}

\newcommand{\HLink}[2]{\hyperref[#2]{#1~\ref*{#2}}}
\newcommand{\HLinkSuffix}[3]{\hyperref[#2]{#1\ref*{#2}{#3}}}

\newcommand{\figlab}[1]{\label{fig:#1}}
\newcommand{\figref}[1]{\HLink{Figure}{fig:#1}}

\newcommand{\lemlab}[1]{\label{lemma:#1}}
\newcommand{\lemref}[1]{\HLink{Lemma}{lemma:#1}}%

\providecommand{\eqlab}[1]{}%
\renewcommand{\eqlab}[1]{\label{equation:#1}}

\providecommand{\remove}[1]{}%

\newcommand{\pth}[1]{\mleft(#1\mright)}%

\newcommand{\ceil}[1]{\mleft\lceil {#1} \mright\rceil}
\newcommand{\floor}[1]{\mleft\lfloor {#1} \mright\rfloor}

\newcommand{\cardin}[1]{\left\lvert {#1} \right\rvert}%

\renewcommand{\th}{th\xspace}

\newlist{compactenumA}{enumerate}{5}%
\setlist[compactenumA]{topsep=0pt,itemsep=-1ex,partopsep=1ex,parsep=1ex,%
   label=(\Alph*)}%

\newlist{compactenuma}{enumerate}{5}%
\setlist[compactenuma]{topsep=0pt,itemsep=-1ex,partopsep=1ex,parsep=1ex,%
   label=(\alph*)}%

\newlist{compactenumI}{enumerate}{5}%
\setlist[compactenumI]{topsep=0pt,itemsep=-1ex,partopsep=1ex,parsep=1ex,%
   label=(\Roman*)}%

\newlist{compactenumi}{enumerate}{5}%
\setlist[compactenumi]{topsep=0pt,itemsep=-1ex,partopsep=1ex,parsep=1ex,%
   label=(\roman*)}%

\newlist{compactitem}{itemize}{5}%
\setlist[compactitem]{topsep=0pt,itemsep=-1ex,partopsep=1ex,parsep=1ex,%
   label=\ensuremath{\bullet}}%

\numberwithin{figure}{section}%
\numberwithin{table}{section}%
\numberwithin{equation}{section}%

\usepackage{mathcalb}

\newcommand{\Root}{\mathcalb{r}}
\newcommand{\Term}[1]{\textsf{#1}}

\newcommand{\DFS}{\Term{DFS}\xspace}%
\newcommand{\nI}{k}%
\newcommand{\curr}{\mathcalb{c}}%
\newcommand{\CTree}{\mathcal{C}}%
\newcommand{\DTree}{\mathcal{T}}%
\newcommand{\Tree}{\mathcal{T}}%

\providecommand{\etal}{et~al.\xspace}
\renewcommand{\etal}{et~al.\xspace}

\newcommand{\LCAY}[2]{\mathrm{lca}\pth{#1,#2}}
\newcommand{\VX}[1]{V\pth{#1}}%

\usepackage{stmaryrd}%
\providecommand{\IntRange}[1]{\mleft\llbracket #1 \mright\rrbracket}
\newcommand{\IRX}[1]{\IntRange{#1}}%
\newcommand{\IRY}[2]{\left\llbracket #1:#2 \right\rrbracket}
\newcommand{\orderX}[1]{\prec_{#1}^{}}
\providecommand{\deflab}[1]{\label{def:#1}}
\providecommand{\deflab}[1]{\label{def:#1}}

\newcommand{\dLCAY}[2]{d_{\mathrm{lca}}(#1,#2)}
\newcommand{\defrefY}[2]{\hyperref[def:#2]{#1}}
\newcommand{\constA}{\alpha}
\newcommand{\UB}{\mathcal{W}}%
\newcommand{\II}{\mathcal{I}}%

\bibliography{explore_tree}%

\begin{document}

\title{Bifurcation: How to Explore a Tree}

\author{Sariel Har-Peled\SarielThanks{Work on this paper was partially
      supported by NSF AF award CCF-2317241.  }}

\date{\today}

\maketitle

\begin{abstract}
    Parametric search is a powerful but hard-to-use technique in computational geometry. Avraham \etal \cite{afkks-dsfds-15} presented a new approach, called \emphw{bifurcation}, that performs faster under certain circumstances. Intuitively, when the underlying decider execution can be rolled back cheaply and the decider has a near-linear running time. For some problems, this leads to fast algorithms that beat the seemingly natural lower bound arising from distance selection.

    Bifurcation boils down to a tree exploration problem. You are given a binary (unfortunately implicit) tree of height $n$ and $\nI$ internal nodes with two children (all other internal nodes have a single child), and assume each node has an associated parameter value. These values are sorted in the inorder traversal of the tree.   Assume there is (say) a node (not necessarily a leaf) that is the target node that the exploration needs to discover.

    The player starts from the root. At each step, the player can move to adjacent nodes to the current location (i.e., one of the children or the parent). Alternatively, the player can call an oracle on the current node, which returns either that it is the target (thus, mission accomplished!\footnote{The player can withdraw,}) or whether the target value is strictly smaller or larger than the current one.

    A naive algorithm explores the whole tree, in $O(n \nI)$ time, then performs $O(\log \nI n)$ calls to the oracle to find the desired leaf. Avraham \etal \cite{afkks-dsfds-15} showed that this can be improved to $O(n \sqrt{\nI} )$ time, and $O( \sqrt{\nI} \log n)$ oracle calls.

    Here, we improve this to $O(n \sqrt{\nI} )$ time, with only $ O( \sqrt{\nI} + \log n)$ oracle calls. We also show matching lower bounds, under certain assumptions. The new algorithm implies a minor improvement (i.e., roughly a $\log n$ factor) in the running times for bifurcation algorithms. Unfortunately, these algorithms are infested with large polylog terms, and this improvement is underwhelming. Nevertheless, we believe our interpretation of bifurcation as a tree exploration problem, and the associated algorithm, are of independent interest.
\end{abstract}

\section{Introduction}

\paragraph{Bifurcation.}
The bifurcation technique was first presented by Avraham \etal \cite{afkks-dsfds-15} (a somewhat more readable presentation is provided by Kaplan \etal \cite{kkss-uwrsp-23}). A nifty recent paper using this technique is by Chan and Huang \cite{ch-farsp-25}. See the introduction of \cite{ch-farsp-25} for a recent literature review.

\paragraph{The input.}
An \emphi{$(n,\nI)$-tree} $\Tree$ is a rooted tree with exactly $\nI$ internal nodes having two children, and all its other interval nodes having only one child. Furthermore, all the leaves of $\Tree$ are at a distance of at most $n$ from the root $\Root$ of $\Tree$. An internal node with two children is a \emphi{fork}. A fork has a \emphi{left} and \emphi{right} child, which readily induces an ordering on the nodes of the tree -- by the \DFS traversal of $\Tree$ that recursively visits the left child before the node itself, and then the right child. A special (unknown) node of $\Tree$ is the \emphi{target}. The task at hand is to discover the target.

The tree is given implicitly -- initially, only the root $\Root$ is given. The algorithm can traverse the tree by walking on it locally. Namely, it maintains a current node $\curr$ in the tree. In a single step, the algorithm can move $\curr$ to any neighboring nodes, either upward (to its parent) or downward to one of its children (if it is a fork, or to its only child otherwise). When the algorithm enters a previously unexplored vertex, it is also given its status: a fork, a regular vertex, or a leaf.

\begin{defn}
    For any node $x \in \Tree$, let $\pi(x)$ denote the unique \emphw{node-to-root} path from $x$ to the root $\Root$.  For two nodes $u,v \in \VX{\Tree}$, $u$ is \emphi{smaller} than $v$, denoted by $u \orderX{\Tree} v$, if $u$ appears before $v$ in the inorder \DFS traversal of $\Tree$.
\end{defn}

The algorithm may consult (for free\footnote{Well, if you were to bribe the oracle with a gift, it would not refuse you. But, really, cash is best in such cases, judging by Greek mythology.}) with an oracle.

\begin{defn}[Target oracle]
    \deflab{oracle}%
    Given a query node $\curr \in \VX{\Tree}$, the \emphi{target oracle} returns one of the following: %
    \begin{compactenumI}
        \smallskip%
        \item  $\curr$ is the target (the algorithm is done),

        \smallskip%
        \item the target is smaller than $\curr$, or

        \smallskip%
        \item the target is larger than $\curr$.
    \end{compactenumI}
\end{defn}
\medskip%
The challenge is to find the target using a minimal number of steps and oracle calls. We describe a minor variant of this algorithm here (with a minor performance improvement).

The ``obvious'' algorithm explores the whole tree first. As an $(n, \nI)$-tree have $O( n \nI)$ nodes. Thus, a straightforward \DFS would take $O(n \nI)$ steps. One can then find the target using $O\bigl( \log (n \nI ) \bigr)$ oracle calls (using binary search).  This algorithm pays upfront for the exploration -- to do better, one needs to mix the exploration with the binary search.

A natural way to do so is to set an upper bound on the depth of exploration -- that is, the \DFS does not continue the search below nodes of a certain depth, thus avoiding overpaying for exploration, an idea we use below. Specifically, in the $i$\th round, one explores the tree up to depth $i n /\sqrt{\nI}$, then uses binary search using the oracle to discover the leaf leading to the target, and continues the search into this leaf in the next round. This algorithm is (essentially) from Avraham \etal \cite{afkks-dsfds-15}.

\paragraph*{Our results.}
We describe a different approach that performs only partial decimation of the nodes in each round, leading to a faster algorithm. We also show matching lower bounds, first under the assumption that an oracle call costs $\Omega(n)$ time, and the second lower bound is under the assumption that the target is at a leaf, and oracle calls are allowed only on leaves. The latter lower bound leads to an interesting, and somewhat counter-intuitive, lower bound of $\Omega( \log^2 n)$, on searching in a complete binary search tree, where all the values are stored in the leaves, and one has to pay for each edge traversed during the search. See \lemref{bogi} for details.

\section{The algorithm}

In the following, all $\log$s are in base $2$. For integers $x < y$, let $\IRY{x}{y} = \{ x, x+1, \ldots, y\}$ and $\IRX{y} = \IRY{1}{y}$.

\subsection{Description of the algorithm}

The algorithm maintains a subtree $\CTree \subseteq \Tree$, of the part of the tree $\Tree$ that was explored so far.

\subsubsection{Basic operations}

A node $u$ of $\CTree$ marked as a \emphw{stub} can not have the target in its subtree $\Tree_u$. In particular, when a node is marked as a stub, all its children are removed.

\paragraph{Trimming.}

For a specified vertex $u \in \VX{\CTree}$, the algorithm consults with the oracle to decide if $u$ is the target. If so, the algorithm is done. Otherwise, consider the path $\pi(u)$ from $\Root$ to $u$. If, according to the oracle, the target is bigger (resp., smaller) than $u$, then the algorithm marks all the direct left (resp., right) children of nodes of $\pi(u)$ in $\CTree$ as being stubs.

\paragraph{Halving.}

If the current tree $\CTree$ becomes too large (say, it has more than $\constA n$ nodes, where $\constA$ is a sufficiently large constant), it is useful to shrink it to roughly half its current size. For a leaf $u \in \CTree$, let $R(u)$ (resp., $L(u)$) be the set of all the vertices of $\CTree$ that are smaller (resp., larger) than $u$. Using \DFS, the algorithm computes, in linear time in the size of $\CTree$, a node $u$, such that $\cardin{\bigl.|L(u)| -|R(u)|} \leq 1$ (i.e., the ``median'' in the inorder of the nodes of $\CTree$). Next, the algorithm applies the trimming operation to $u$, resulting in an updated $\CTree$ having roughly half of its original size -- more precisely, if $\CTree$ had $m \geq \constA n$ nodes, the new tree has at most $1 + (\constA /2+ 1)n$ nodes (as the nodes on $\pi(u)$ are preserved). This algorithm requires $O( \cardin{\CTree})$ time/steps, and one oracle call.

A straightforward modification of this procedure leads to an alternative version that halves the number of leaves in $\CTree$ (as a reminder, stubs are not considered to be leaves).

\subsubsection{The algorithm in detail}

Let $\psi$ be a prespecified parameter -- roughly, the number of oracle calls the algorithm would perform. For simplicity of exposition, assume the quantities
\begin{equation*}
    L = \frac{\nI}{\psi},
    \qquad%
    \Delta =\frac{2n}{\psi},
    \qquad\text{and}\qquad%
    \UB = L \Delta = \frac{2\nI n}{\psi^2}.
\end{equation*}
are integers. Initially, the explored tree $\CTree_0$ contains only the root $\Root$ of $\Tree$.

\paragraph{The $i$\th round.}
In the beginning of the $i$\th round, for $i = 1, \ldots, \psi/2$, the algorithm starts at the root of $\CTree_{i-1}$, and performs \DFS up to depth $D_i = i \Delta$. The \DFS does not continue through stubs but continues into the unexplored parts of $\Tree$ when reaching leaves of $\CTree_{i-1}$ at depth smaller than $D_i$. Let $\DTree_i$ be the resulting tree. All the leaves of $\DTree_i$ are at depth $D_i$. The algorithm now repeatedly performs halving on $\DTree_i$ till it has at most $\UB$ nodes and at most $L$ leaves. Let $\CTree_i$ be the resulting tree.

\paragraph{After the final round.}
The above results in a tree $\CTree_{\psi}$, that might have up to $\UB$ nodes (and $\nI+1$ leaves). The target is found via a binary search on these nodes, performing $O( \log \UB )$ oracle calls.

\subsection{Analysis}

\begin{lemma}
    In the $i$\th round, the total number of oracle calls is $o_i = O( 1 + f_i /L )$, and the total number of steps is $O\bigl( \UB + f_i \Delta )$, where $f_i$ is the number of new forks discovered in this round.
\end{lemma}
\begin{proof}
    During the computation of $\DTree_i$, each edge of $\CTree_{i-1}$ (which forms the upper part of $\DTree_i$) is traversed at most twice by the \DFS process. Thus, the work involved in the \DFS over $\CTree_{i-1}$ is bounded by $O(\UB)$, as the algorithm maintains the invariant that the tree computed in each round has size $O(\UB)$ and at most $L$ leaves.

    Each new fork discovered during the $i$\th round can contribute an additional path of length $\Delta$ to $\Tree_i$ till the \DFS either reaches a new fork, or the bottom level $D_i$. Thus, the total size of $\DTree_i$, and thus the total amount of work spent on exploration in this round, is bounded by
    \begin{equation*}
        n_i = O( \UB +
        \#_{\mathrm{leafs}}(\CTree_{i-1}) \Delta + f_i \Delta)
        =
        O\bigl( \UB + (L + f_i) \Delta )
        =
        O\bigl( \UB +  f_i \Delta ),
    \end{equation*}
    since $\UB = \Delta L$.  The halving, done at the end of the round, requires at most
    \begin{equation*}
        o_i
        =
        O\pth{1 + \frac{n_i}{\UB} + \frac{f_i}{L} }
        =%
        O\pth{1 + \frac{ (L + f_i) \Delta}{\Delta L} + \frac{f_i}{L} }
        =%
        O\pth{1 +  \frac{f_i}{L} }
    \end{equation*}
    oracle calls.
\end{proof}

\begin{theorem}
    Given an implicit $(n,\nI)$-tree containing an unknown target, and a parameter $\psi \in \IRX{\nI}$, the target can be found in $O(n \nI /\psi )$ time, using $ O( \psi + \log n)$ calls to an \defrefY{oracle}{oracle}.
\end{theorem}
\begin{proof}
    The number of oracle calls done before the final cleanup stage is
    \begin{equation*}
        \sum_{i=1}^{\psi} o_i
        =%
        \sum_{i=1}^\psi O\bigl( 1 + f_i /L \bigr)
        =%
        O( \psi +  \nI/L )
        =
        O( \psi )
    \end{equation*}
    since $\sum_i f_i = \nI$ and $L = {\nI}/{\psi}$. The final cleanup requires an additional $O( \log \UB ) = O( \log n)$ oracle calls. The total number of steps/work is
    \begin{equation*}
        \sum_{i=1}^{\psi} n_i
        =%
        \sum_{i=1}^\psi
        O\bigl( \UB +  f_i \Delta )
        =%
        O( \psi \UB +  \nI \Delta )
        =%
        O\pth{
           \frac{\nI n}{\psi} },
    \end{equation*}
    as $\UB = {\nI n} / {\psi^2}$ and $\Delta = {n} / {\psi}$.
\end{proof}

\begin{corollary}
    Given an implicit $(n,\nI)$-tree containing an unknown target node, the target can be found in $O(n \sqrt{\nI} )$ time, using $ O( \sqrt{\nI} + \log n)$ calls to the oracle.
\end{corollary}

\begin{remark}
    A natural variant of the above is when all the leaves are at depth $n$, the target must be a leaf, and the oracle queries are performed only on the leaves. In such a case, one can replace the $O(\log n)$ by $O(\log \nI)$ in the number of oracle queries, but that does not seem to be significant\footnote{Not clear anything in this writeup is significant.}.
\end{remark}

\section{Lower bounds}

\subsection{When the decider takes linear time}

\begin{lemma}
    Assume that a call to the decider takes linear time (i.e.,
    $O(n)$). Then, in the worst case, an algorithm exploring an
    $(n,\nI)$ tree must take $O( n \sqrt{\nI})$ time.
\end{lemma}

\begin{proof}
    Consider the ``complete'' binary tree of height $h=\sqrt{\nI}$, where we replace each edge by a path of length $\Delta = n/\sqrt{\nI}$ (i.e., it is made out of $\Delta$ edges). Let $\Tree$ denote this tree.

    The adversary strategy is now the following -- whenever the oracle is called, it returns the answer that maximizes the number of vertices of the tree $\Tree$ that might contain the target. If the player had exposed $\nI$ forks, then the adversary amends all the undiscovered forks in $\Tree$ into regular nodes by deleting one of their children (thus, after this ``trimming'' $\Tree$ has size at most $O(n \nI)$).

    Now, if the player forced the adversary to trim $\Tree$, then it already spent $\Omega( \nI \Delta ) = \Omega(n \sqrt{\nI})$ time on exploration as it had to traverse at least $\nI-1$ edges, each one of length $\Delta$. Otherwise, the standard lower bound on binary search over $L = 2^{\sqrt{\nI}}$ elements, implies the algorithm must have issues $\Omega( \log L ) = \Omega(\sqrt{\nI})$ oracle calls, which requires $\omega( \sqrt{\nI} n )$ time, as claimed.
\end{proof}

\subsection{When the oracle can be consulted only in leaves}

An interesting variant is that an oracle call still costs only $O(1)$ time, but one can perform such calls on leaves of depth $n$ (here we assume the target is also a leaf of depth $n$).

\subsubsection{Searching for a needle in a complete binary tree}

Consider a complete binary tree $\Tree$ of height $h$. There are $\nu = 2^{h}$ leafs in $\Tree$, and let $i \in \II = \IRY{0}{\nu-1}$ be the label associated with the $i$\th leaf in the inorder traversal of $\Tree$. We label the $i$\th leaf with the binary string $B = \{0,1\}^h$ whose binary value is $i$ (i.e., the binary string encodes the path from the root to this leaf). The target $t$ is one of these leaves, and an oracle query must be one of the leaves of $\Tree$.

The player asks a sequence of queries, where a query is a label in $\II$. For a query $q$, the adversary either answers that $q \prec t$, $q \succ b$, or $q =t$ (i.e., the target $t$ was found). After the $i$\th iteration, there is an \emphw{active} range $R_i = \IRY{x_i}{y_i}$ where $t$ might be, where $R_1 = \II$.

\paragraph{The pricing model.} %
The \emphi{rank} of a node $x \in \VX{\Tree}$ is the height of its subtree. Thus, a leaf of $\Tree$ has rank zero, while the root $\Root$ has rank $h$. Given two strings $p, q \in B$, their \emphw{LCA distance} is the rank of their $\LCAY{p}{q}$, and it is denoted by $\dLCAY{q}{q'}$.

Given a sequence of queries $Q = \{q_1, \ldots, q_{k}\}$, let $Q_i = \{ q_1, \ldots, q_i \}$ denote the $i$\th prefix. The set of edges in the tree that must be traversed before one can use the oracle is the set
\begin{equation*}
    E(Q) = \cup_{q \in Q} \pi(q).
\end{equation*}
In particular, the \emphi{price} of the $i$\th query is %
\begin{equation*}
    |E(Q_i)| - |E(Q_{i-1})|
    =%
    \dLCAY{q_i}{Q_{i-1}} = \smash{\min_{p \in Q_{i-1}}} \dLCAY{q}{p}.
\end{equation*}
And the overall price of $Q$ is
\begin{math}
    |E(Q)| = \sum\nolimits_{i=1}^k \dLCAY{q_i}{Q_{i-1}}.
\end{math}

Here, we show that the optimal strategy has a price $\Omega(h^2)$ in the worst case.

\paragraph{The adversary strategy.}
In the $i$\th round, if the query $\beta_i\in \II$ is outside the current active range $R_{i} = \IRY{x_{i}}{\beta_{i}}$, the adversary returns the relation between $\beta_i$ and (say) $x_{i}$ (and sets $R_{i+1} \leftarrow R_{i}$).  Otherwise, the query $\beta_i \in R_{i}$, and it breaks $R_{i}$ into two ranges
\begin{equation*}
    A_i = \IRY{x_{i}}{\beta_i -1}%
    \qquad\text{and}\qquad%
    B_i = \IRY{\beta_i +1}{y_{i}}.
\end{equation*}
The adversary answers ``$ t\prec \beta_i$'' if $\cardin{A_i} > \cardin{B_i}$ (i.e., $R_i \leftarrow A_i$), and answers ``$t \succ \beta_i$'' (i.e., $R_{i+1} \leftarrow B_i$) otherwise.  The game ends when the target is isolated, that is, $\cardin{R_{i+1}}=1$.

For the simplicity of exposition, we assume the player is aware of the adversary strategy, so it is not going to waste time on useless queries -- such as a query outside the active range, or at the start/end values in the current active range if it has more than two values (both assumptions can be removed with some tedium).

\paragraph{Analysis.}

\begin{figure}[h!]
    \centerline{\includegraphics{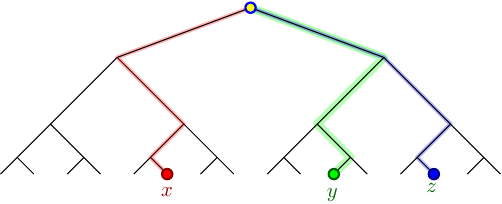}}
    \caption{}
    \figlab{path:trap}
\end{figure}
\begin{lemma}
    \lemlab{trap}%
    Consider three strings $x,y,z \in B$, with $x<y<z$, where first the $\pi(x) \cup \pi(y)$ are queries were already performed (i.e., $Q = \{ x,z,y\}$). Then, the price of $y$ is $\rho = \dLCAY{y}{\{x,z\}}$, where
    \begin{math}
        \rho \geq \min\Bigl( \log(y-x), \log(z-y) \Bigr).
    \end{math}
    and also $\rho \geq 1$.
\end{lemma}
\begin{proof}
    Assume the situation is as depicted in \figref{path:trap}, and $\rho = \dLCAY{y}{z} < \dLCAY{x}{y}$, as the other case follows by symmetry. Let $\alpha = z' - y'$. For a node of rank $\rho$, the distance between its two extreme leaves is $2^\rho-1$. Thus,
    \begin{equation*}
        2^\rho - 1 \geq \alpha
        \implies
        \rho \geq
        \ceil{\log( \alpha+1)} = 1 + \floor{\log \alpha}
        \geq \log \alpha.
    \end{equation*}
\end{proof}

\begin{lemma}
    \lemlab{bogi}%
    Consider a complete binary search tree $\Tree$ of height $h$. Against an adaptive adversary, when queries can be issued only at the leaves, any game that ends with the target found has a price of at least $\Omega(h^2)$.
\end{lemma}
\begin{proof}
    By induction on the length of the active range. Let $P(\Delta)$ be the minimum price of any game carried out once the active range has $\Delta$ elements. The claim is that $P( \Delta ) \geq c \log^2 \Delta$, where $c > 0$ is some constant, and $\Delta > 1$.  The claim clearly holds for $\Delta \leq 10$, so assume $\Delta > 10$.

    So, the game is at the beginning of the $i$\th round, and let $R_{i}$ be the active range with $|R_{i}| = \Delta$. Let $j$ be the minimum index such that $\cardin{R_j} \leq \cardin{R_{i}}/2$ (observe that $\cardin{R_j} > \Delta/4$).  Let $q_{i}, \ldots, q_{j-1}$ be the queries used in these rounds -- we assume for simplicity of exposition that $q_k$ is in the interior of $R_k$. Let $J_{i}, \ldots, J_{j-1}$ be the ranges thrown away. Specifically, $R_{k+1} = R_{k} \setminus J_{k}$, and $q_{k}$ is an endpoint of $J_{k}$, for all $k$. By the adversarial behavior, we have that $\cardin{J_{k}} \leq \cardin{R_{k+1}} + 1$ (the $+1$ is for the case where $|R_k|$ is odd, and $q_k$ is the median).

    For simplicity of exposition, assume that $0$ and $2^h-1$ are not in $R_u$, for all $u \in \IRY{i}{j}$. The range $R_{k}=\IRY{x_{k}}{y_{k}}$ was created because the player asked the queries $x_{k}-1$ and $y_{k}+1$ in some earlier rounds. Thus, the price of $q_k$ is $\xi_k = \dLCAY{q_{k}}{\{x_{k}-1, y_{k}+1\}}$, as the root path $\pi(q_k)$ is ``trapped'' between $\pi(x_{k}-1)$ and $\pi(y_{k}+1)$.  Thus, by \lemref{trap}, $\xi_k \geq \log \cardin{J_{k+1}} \geq 1$, since $\cardin{J_{k+1}} \geq 2$.  We have
    \begin{equation*}
        P(\Delta)
        \geq
        \sum_{k=i}^{j-1} \xi_k
        + P(|R_j|)
        \geq
        \log \prod_{k=i+1}^{j} \cardin{J_k}
        + P(\Delta/4).
    \end{equation*}
    The latter is minimized by taking the refuse intervals $J_k$ to be as large as possible (and minimizing their overall number). In particular, setting $|J_{i+1}|$ to its maximum possible value $\ceil{\Delta/2}$, we have
    \begin{equation*}
        P(\Delta)
        \geq
        \log \frac{\Delta}{2}
        + P(\Delta/4),
    \end{equation*}
    which readily implies that $P(\Delta) = \Omega(\log^2 \Delta)$.
\end{proof}

\paragraph{Lower bound for the $(n,\nI)$-tree case}

We assume $n > \nI$. The idea is to create an $(n,\nI)$-tree, by taking the complete binary tree of height $h = \sqrt{\nI}$, and replacing an edge by a path of length $\nabla = n/\sqrt{\nI}$. Clearly, the resulting tree is an $(n,\nI)$-tree, but it has way too many forks (i.e., $2^{\sqrt{\nI}}$ instead of $\nI$). The key insight is that this does not matter -- as soon as the player visits $\nI$ forks, the adversary has already won, and it can pretend that the rest of the unexplored tree never existed.

We change the charging scheme in the above for exploring $(n, \nI)$-tree -- the player has to pay for all the steps made along a path, only when it arrives at a fork, and then it has to pay for all the vertices visited on this path. The user's strategy here can be interpreted as one applied to a complete binary tree. \lemref{bogi} implies that this strategy must visit $\Omega(h^2) = \Omega(\nI)$ forks. Still, the total price of the edges arriving at these forks is $\Omega(h^2 \nabla ) = \Omega( n h ) = \Omega( \sqrt{\nI}n)$.

\begin{theorem}
    Any algorithm that finds an unknown target in an $(n,\nI)$-tree, and issues queries only at the leaves, must make $\Omega( \sqrt{\nI} n )$ steps, in the worst case.
\end{theorem}

\paragraph*{Acknowledgments.}
The author thanks Timothy Chan and Micha Sharir for insightful discussions on this problem.

\printbibliography

\end{document}